\documentclass[envcountsame]{llncs}

\usepackage{amsmath,amssymb,url}
\usepackage[T1]{fontenc}
\usepackage[utf8]{inputenc}
\usepackage[final]{graphicx}

\author{Zolt\'an \'Esik\inst{1}\fnmsep%
  \thanks{Partially supported by grant no.~K~75249 from the National
    Foundation of Hungary for Scientific Research and by the
    T\'AMOP-4.2.2/08/1/2008-0008 program of the Hungarian National
    Development Agency.}
  \and Andreas Maletti\inst{2}\fnmsep%
  \thanks{Financially supported by the \emph{Ministerio de Educaci\'on
      y Ciencia} (MEC) grant JDCI-2007-760 and the \emph{European
      Science Foundation} (ESF) short-visit grant~2978 in the activity
    ``Automata: from Mathematics to Applications''.}}
\title{Simulations of Weighted Tree Automata}
\institute{University of Szeged, Department of Computer Science \\
  \'Arp\'ad t\'er 2, 6720 Szeged, Hungary \\
  email: \url{ze@inf.u-szeged.hu} \\[0.5ex]
  \and Universitat Rovira i Virgili, Departament de Filologies
  Rom\`aniques \\ Avinguda de Catalunya 35, 43002 Tarragona, Spain \\
  email: \url{andreas.maletti@urv.cat}}

\DeclareMathOperator{\rk}{rk}

\providecommand*{\nat}[0]{\ensuremath{\bbbn}}
\providecommand*{\integer}[0]{\ensuremath{\bbbz}}
\providecommand*{\abs}[1]{\ensuremath{\lvert #1 \rvert}}
\providecommand*{\sem}[1]{\ensuremath{\lVert #1 \rVert}}
\providecommand*{\seq}[3]{\ensuremath{#1_{#2}, \dotsc, #1_{#3}}}
\providecommand*{\word}[3]{\ensuremath{#1_{#2} \dotsm #1_{#3}}}
\providecommand*{\series}[2]{\ensuremath{#1 \langle\!\langle T_{#2}
    \rangle\!\rangle}} 

\begin{document}
\maketitle

\begin{abstract}
  Simulations of weighted tree automata~(wta) are considered.  It is
  shown how such simulations can be decomposed into simpler functional
  and dual functional simulations also called forward and backward
  simulations.  In addition, it is shown in several cases (fields,
  commutative rings, \textsc{Noetherian} semirings, semiring of
  natural numbers) that all equivalent wta $M$~and~$N$ can be joined
  by a finite chain of simulations.  More precisely, in all mentioned
  cases there exists a single wta that simulates both $M$~and~$N$.  Those
  results immediately yield decidability of equivalence provided that
  the semiring is finitely (and effectively) presented.
\end{abstract}

\section{Introduction}
\label{sec:Intro}
Weighted tree automata (or equivalently, weighted tree grammars) are
widely used in applications such as model
checking~\cite{abdjonmahors02} and natural language
processing~\cite{knigra05}.  They finitely represent mappings, called
tree series, that assign a weight (taken from a semiring) to each
tree.  For example, a probabilistic parser would return a tree series
that assigns to each parse tree its likelihood.  Consequently, several
toolkits~\cite{klamol01,maykni06,cle08} implement weighted tree
automata.

The notion of simulation that is used in this paper is a
generalization of the simulations for unweighted and weighted (finite)
string automata of~\cite{bloesi93,esikui01}.  The aim is to relate
structurally equivalent automata.  The results
of~\cite[Section~9.7]{bloesi93} and \cite{koz94} show that two
unweighted string automata (i.e., potentially nondeterministic string
automata over the \textsc{Boolean} semiring) are equivalent if and
only if they can be connected by a finite chain of relational
simulations, and that in fact \emph{functional} and \emph{dual
  functional} simulations are sufficient.  Simulations for weighted
string automata~(wsa) are called \emph{conjugacies}
in~\cite{bealomsak05,bealomsak06}, where it is shown that for all
fields, many rings including the ring~$\integer$ of integers, and the
semiring~$\nat$ of natural numbers, two wsa are equivalent if and only
if they can be connected by a finite chain of simulations.  It is also
shown that even a finite chain of functional (\emph{covering}) and
dual functional (\emph{co-covering}) simulations is sufficient.  The
origin of those results can be traced back to the pioneering work of
\textsc{Sch\"utzenberger} in the early 60's, who proved that every wsa
over a field is equivalent to a minimal wsa that is simulated by every
\emph{trim} equivalent wsa~\cite{berreu84}.  Relational simulations of
wsa are also studied in~\cite{buc08}, where they are used to reduce
the size of wsa.  The relationship between functional simulations and
the \textsc{Milner}-\textsc{Park} notion of
bisimulation~\cite{mil80,par81} is discussed in~\cite{bloesi93,buc08}.

In this contribution, we investigate simulations for weighted (finite)
tree automata~(wta).  \textsc{Sch\"utzenberger}'s minimization method
was extended to wta over fields in~\cite{aleboz89,boz91}.  In
addition, relational and functional simulations for wta are probably
first used in~\cite{esi98,esi10b,hogmalmay07d}.  Moreover, simulations
can be generalized to presentations in algebraic
theories~\cite{bloesi93}, which seems to cover all mentioned
instances.  Here, we extend the results
of~\cite{bealomsak05,bealomsak06} to wta.  In particular, we show that
two wta over a ring, \textsc{Noetherian} semiring, or the
semiring~$\nat$ are equivalent if and only if they are connected by a
finite chain of simulations.  Moreover, we discuss when the
simulations can be replaced by functional and dual functional
simulations, which are efficiently computable~\cite{hogmalmay07d}.
Such results are important because they immediately yield decidability
of equivalence provided that the semiring is finitely and effectively
presented.

\section{Preliminaries}
\label{sec:Prelim}
The set of nonnegative integers is~$\nat$.  For every $k \in \nat$,
the set $\{i \in \nat \mid 1 \leq i \leq k\}$ is simply denoted
by~$[k]$. We write $\abs A$ for the cardinality of the set~$A$.  A
\emph{semiring} is an algebraic structure~${\cal A} = (A, \mathord+,
\mathord\cdot, 0, 1)$ such that $(A, \mathord+, 0)$ and $(A,
\mathord\cdot, 1)$ are monoids, of which the former is commutative,
and $\cdot$~distributes both-sided over finite sums (i.e., $a \cdot 0
= 0 = 0 \cdot a$ for every $a \in A$ and $a \cdot (b + c) = ab + ac$
and $(b + c) \cdot a = ba + ca$ for every $a, b, c \in A$).  The
semiring~$\mathcal A$ is \emph{commutative} if $(A, \mathord\cdot, 1)$
is commutative.  It is a \emph{ring} if for every $a \in A$ there
exists an \emph{additive inverse}~$-a \in A$ such that $a + (-a) = 0$.
The set~$U$ is the set $\{ a \in A \mid \exists b \in A \colon ab = 1
= ba\}$ of \emph{(multiplicative) units}.  The semiring~${\cal A}$ is
a \emph{semifield} if $U = A \setminus \{0\}$; i.e., for every $a \in
A$ there exists a \emph{multiplicative inverse}~$a^{-1} \in A$ such
that $aa^{-1} = 1 = a^{-1}a$.  A \emph{field} is a semifield that is
also a ring.  For every $B \subseteq A$ let $\langle
B\rangle_{\mathord+} = \{ b_1 + \dotsb + b_n \mid n \in \nat, \seq b1n
\in B\}$. If $A = \langle B\rangle_{\mathord+}$, then ${\cal A}$~is
\emph{additively generated by~$B$}.  Finally, it is
\emph{equisubtractive} if for every $a_1, a_2, b_1, b_2 \in A$ with
$a_1 + b_1 = a_2 + b_2$ there exist $c_1, c_2, d_1, d_2 \in A$ such
that (i)~$a_1 = c_1 + d_1$, (ii)~$b_1 = c_2 + d_2$, (iii)~$a_2 = c_1 +
c_2$, and (iv)~$b_2 = d_1 + d_2$.

The semiring~${\cal A}$ is \emph{zero-sum free} if $a + b = 0$ implies
$0 \in \{a, b\}$ for every $a, b \in A$.  Clearly, any nontrivial
(i.e., $0 \neq 1$) ring is not zero-sum free.  Moreover, ${\cal A}$~is
\emph{zero-divisor free} if $a \cdot b = 0$ implies $a = 0 = b$ for
every $a, b \in A$.  All semifields are trivially zero-divisor free.
Finally, the semiring~${\cal A}$ is \emph{positive} if it is zero-sum
free and zero-divisor free.  An infinitary sum
operation~$\mathord{\sum}$ is a family~$(\mathord{\sum_I})_I$ such
that $\mathord{\sum_I} \colon A^I \to A$.  We generally write $\sum_{i
  \in I} a_i$ instead of $\sum_I (a_i)_{i \in I}$.  The
semiring~${\cal A}$ together with the infinitary sum
operation~$\mathord{\sum}$ is
\emph{complete}~\cite{eil74,hebwei98,gol99,kar04} if
\begin{itemize}
\item $\sum_{i \in \{j_1, j_2\}} a_i = a_{j_1} + a_{j_2}$ for
  all $j_1\neq j_2$ and $a_{j_1}, a_{j_2} \in A$,
\item $\sum_{i \in I} a_i = \sum_{j \in J} \bigl( \sum_{i \in I_j} a_i
  \bigr)$ for every index set~$I$, partition $(I_j)_{j \in J}$ of~$I$,
  and $(a_i)_{i \in I} \in A^I$, and
\item $a \cdot \bigl(\sum_{i \in I} a_i \bigr) = \sum_{i \in I} aa_i$
  and $\bigl(\sum_{i \in I} a_i \bigr) \cdot a = \sum_{i \in I} a_ia$
  for every $a \in A$, index set~$I$, and $(a_i)_{i \in I} \in A^I$.
\end{itemize}
An ${\cal A}$-\emph{semimodule} is a commutative monoid $(B,
\mathord{+}, 0)$ together with an action $\mathord{\cdot} \colon A
\times B \to B$, written as juxtaposition, such that for every $a, a'
\in A$ and $b, b' \in B$
\begin{itemize}
\item $(a + a') b = ab + a'b$ and $a(b + b') = ab + ab'$,
\item $0 b = 0 = a 0$, $1 b = b$ and $(a \cdot a')b = a(a'b)$.
\end{itemize}
The semiring~${\cal A}$ is \textsc{Noetherian} if all subsemimodules
of every finitely-generated ${\cal A}$-semimodule are again
finitely-generated.

In the following, we often identify index sets of the same
cardinality.  Let $X \in A^{I_1 \times J_1}$ and $Y \in A^{I_2 \times
  J_2}$ for some finite sets~$I_1, I_2, J_1, J_2$.  We use upper-case
letters (like~$C$, $D$, $E$, $X$, $Y$) for matrices and the
corresponding lower-case letters for their entries.  A matrix $X \in
A^{I \times J}$ is \emph{relational} if $x_{ij} \in \{0, 1\}$ for
every $i \in I$ and $j \in J$.  Clearly, a relational matrix defines a
relation $\rho_X \subseteq I \times J$ by $(i, j) \in \rho_X$ if and
only if $x_{ij} = 1$ (and vice versa).  Moreover, we call a relational
matrix \emph{functional}, \emph{surjective}, or \emph{injective} if
its associated relation has this property.  As usual, we denote the
\emph{transpose} of a matrix~$X$ by~$X^{\mathrm T}$, and we call
$X$~\emph{nondegenerate} if its has no rows or columns of entirely
zeroes.  A \emph{diagonal} matrix~$X$ is such that $x_{ij} = 0$ for
every $i \neq j$.  Finally, the matrix~$X$ is invertible if there
exists a matrix~$Y$ such that $XY = I = YX$ where $I$~is the unit
matrix.  The \textsc{Kronecker} product $X \otimes Y \in A^{(I_1
  \times I_2) \times (J_1 \times J_2)}$ is such that $(X \otimes
Y)_{(i_1, i_2), (j_1, j_2)} = x_{i_1, j_1} y_{i_2, j_2}$ for every
$i_1 \in I_1$, $i_2 \in I_2$, $j_1 \in J_1$, and $j_2 \in J_2$.
Clearly, the \textsc{Kronecker} product is, in general, not
commutative and $(1) \in A^{[1]}$ acts both-sided as neutral element.
We let $X^{0,\mathord{\otimes}} = (1)$ and $X^{i+1, \mathord{\otimes}}
= X^{i, \mathord{\otimes}} \otimes X$ for every $i \in \nat$.

Finally, let us move to trees.  A \emph{ranked alphabet} is a finite
set~$\Sigma$ together with a mapping~$\mathord{\rk} \colon \Sigma \to
\nat$.  We often just write~$\Sigma$ for a ranked alphabet and assume
that the mapping~$\rk$ is implicit.  We write $\Sigma_k = \{\sigma \in
\Sigma \mid \rk(\sigma) = k\}$ for the set of all $k$-ary symbols.
The set of $\Sigma$-\emph{trees} is the smallest set~$T_\Sigma$ such
that $\sigma(\seq t1k) \in T_\Sigma$ for all $\sigma \in \Sigma_k$ and
$\seq t1k \in T_\Sigma$.  A \emph{tree series} is a mapping $\varphi
\colon T_\Sigma \to A$.  The set of all such tree series is denoted
by~$\series A\Sigma$.  For every $\varphi \in \series A\Sigma$ and $t
\in T_\Sigma$, we often write~$(\varphi, t)$ instead of~$\varphi(t)$.
Let ${\scriptstyle \Box}$ be a distinguished nullary symbol such that
${\scriptstyle \Box} \notin \Sigma$. A $\Sigma$-\emph{context}~$c$ is
a tree of~$T_{\Sigma \cup \{{\scriptstyle \Box}\}}$, in which the
symbol~${\scriptstyle \Box}$ occurs exactly once.  The set of all
$\Sigma$-contexts is denoted by~$C_\Sigma$. For every $c \in C_\Sigma$
and $t \in T_\Sigma$, we write~$c[t]$ for the $\Sigma$-tree obtained
by replacing the unique occurrence of~${\scriptstyle \Box}$ in~$c$
by~$t$.

A \emph{weighted tree automaton (over~${\cal A}$)}, for short: wta, is
a system $(\Sigma, Q, \mu, F)$ with
\begin{itemize}
\item an input ranked alphabet~$\Sigma$,
\item a finite set~$Q$ of \emph{states},
\item transitions $\mu = (\mu_k)_{k \in \nat}$ such that $\mu_k \colon
  \Sigma_k \to A^{Q^k \times Q}$ for every $k \in \nat$, and
\item a \emph{final weight} vector $F \in A^Q$.
\end{itemize}
Next, let us introduce the semantics~$\sem M$ of~$M$.  We first define
the function $h_\mu \colon T_\Sigma \to A^Q$ for every $\sigma \in
\Sigma_k$ and $\seq t1k \in T_\Sigma$ by \[ h_\mu(\sigma(\seq t1k)) =
\bigl( h_\mu(t_1) \otimes \dotsm \otimes h_\mu(t_k) \bigr) \cdot
\mu_k(\sigma) \enspace, \] where the final product~$\cdot$ is the
classical matrix product.  Then $(\sem M, t) = h_\mu(t) F$ for every
$t \in T_\Sigma$, where the product is the usual inner (dot)
product.  

Let $f \colon A \to \{0, 1\}$ be such that $f(0) = 0$ and $f(a) = 1$
for all $a \in A \setminus \{0\}$. The Boolean wta~$f(M)$ (i.e.,
essentially an unweighted tree automaton) corresponding to~$M$ is
$(\Sigma, Q, \mu', F')$ where
\begin{itemize}
\item $\mu'_k(\sigma)_{w, q} = f(\mu_k(\sigma)_{w, q})$ for every
  $\sigma \in \Sigma_k$, $w \in Q^k$, and $q \in Q$, and
\item $F'(q) = f(F(q))$ for every $q \in Q$.
\end{itemize}
The wta~$M$ is \emph{trim} if every state is accessible and
co-accessible in~$f(M)$. In other words, the wta~$M$ is trim if
$f(M)$ is trim.

\section{Simulation}
\label{sec:Sim}
Simulations of automata were defined in~\cite{bloesi93,esikui01} in
order to provide a structural characterization of equivalent
automata.  We will essentially follow the presentation
of~\cite{bealomsak05} here.

\begin{definition}
  \label{df:Conj}
  Let $M = (\Sigma, Q, \mu, F)$ and $N = (\Sigma, P, \nu, G)$ be wta.
  Then $M$~\emph{simulates}~$N$ if there exists a matrix~$X \in A^{Q
    \times P}$ such that
  \begin{itemize}
  \item[(i)] $F = XG$, and
  \item[(ii)] $\mu_k(\sigma) X = X^{k, \mathord{\otimes}} \cdot
    \nu_k(\sigma)$ for every $\sigma \in \Sigma_k$.
  \end{itemize}
  The matrix~$X$ is called \emph{transfer matrix}, and we write $M
  \stackrel X\to N$ if $M$~simulates~$N$ with transfer matrix~$X$.
\end{definition}

\begin{figure}
  \centering
  \includegraphics{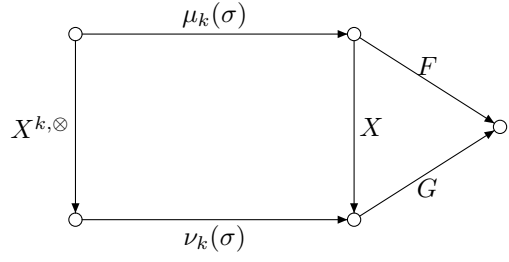}
  \caption{Illustration of simulation.}
  \label{fig:Conj}
\end{figure}

Note that $X^{k, \mathord{\otimes}}_{\word i1k, \word j1k} =
\prod_{\ell = 1}^k x_{i_\ell, j_\ell}$.  We illustrate
Definition~\ref{df:Conj} in Fig.~\ref{fig:Conj}.  If $M \stackrel X\to
M'$ and $M' \stackrel Y\to N$, then $M \stackrel{XY}\to N$.  Thus,
simulations define a preorder on wta.

\begin{theorem}
  \label{thm:Equiv}
  If $M$~simulates~$N$, then $M$~and~$N$ are equivalent.
\end{theorem}

\begin{proof}
  Let $M = (\Sigma, Q, \mu, F)$ and $N = (\Sigma, P, \nu, G)$, and let
  $X \in A^{Q \times P}$ be a transfer matrix.  We claim that
  $h_\mu(t)X = h_\nu(t)$ for every $t \in T_\Sigma$.  We prove this by
  induction on~$t$.  Let $t = \sigma(\seq t1k)$ for some $\sigma \in
  \Sigma_k$ and $\seq t1k \in T_\Sigma$.
  \begin{align*}
    &\phantom{{}={}} h_\mu(\sigma(\seq t1k)) X = \bigl( h_\mu(t_1)
    \otimes \dotsm \otimes h_\mu(t_k) \bigr) \cdot \mu_k(\sigma) X \\
    &= \bigl( h_\mu(t_1) \otimes \dotsm \otimes h_\mu(t_k) \bigr)
    \cdot X^{k, \mathord{\otimes}} \cdot \nu_k(\sigma) = \bigl(
    h_\mu(t_1)X \otimes \dotsm \otimes h_\mu(t_k)X \bigr) \cdot
    \nu_k(\sigma) \\
    &= \bigl( h_\nu(t_1) \otimes \dotsm \otimes h_\nu(t_k) \bigr)
    \cdot \nu_k(\sigma) = h_\nu(\sigma(\seq t1k))
  \end{align*}
  With this claim, the statement can now be proved easily.  For every
  $t \in T_\Sigma$
  \[ (\sem M, t) = h_\mu(t)F = h_\mu(t)XG = h_\nu(t)G = (\sem N, t)
  \enspace. \tag*{\qed} \]
\end{proof}

\begin{lemma}
  \label{lm:Trim}
  Let $M$~and~$N$ be trim wta such $M \stackrel X\to N$.  If (i)~$X$
  is functional or (ii)~${\cal A}$ is positive, then $X$~is
  nondegenerate.
\end{lemma}

\begin{proof}
  Let $M = (\Sigma, Q, \mu, F)$ and $N = (\Sigma, P, \nu, G)$.
  Moreover, let \[ J = \{ p \in P \mid \forall q \in Q \colon x_{qp} =
  0\} \enspace. \] Then $\nu_k(\sigma)_{w, j} = 0$ for every $\sigma
  \in \Sigma_k$, $w \in (P \setminus J)^k$, and $j \in J$. This is
  seen as follows.  Since $\mu_k(\sigma)X = X^{k, \mathord{\otimes}}
  \cdot \nu_k(\sigma)$ we obtain
  \begin{align}
    \label{eq:Trim1}
    \sum_{q \in Q} \mu_k(\sigma)_{\word q1k, q} \cdot x_{qj} = 0 &=
    \sum_{\seq p1k \in P} \Bigl( \prod_{\ell = 1}^k x_{q_\ell,
      p_\ell} \Bigr) \cdot \nu_k(\sigma)_{\word p1k, j}
  \end{align}
  for every $\seq q1k \in Q$ and $j \in J$.  If $X$~is functional, then
  \[ \sum_{\seq p1k \in P} \Bigl( \prod_{\ell = 1}^k x_{q_\ell,
    p_\ell} \Bigr) \cdot \nu_k(\sigma)_{\word p1k, j} =
  \nu_k(\sigma)_{\rho_X(q_1) \dotsm \rho_X(q_k), j} = 0 \enspace, \]
  which proves the claim.  On the other hand, if ${\cal A}$~is
  positive, then \eqref{eq:Trim1}~implies that $\prod_{\ell = 1}^k
  x_{q_\ell, p_\ell} \cdot \nu_k(\sigma)_{\word p1k, j} = 0$ for every
  $\seq p1k \in P$.  Since for every $p_\ell \notin J$, there exists
  $q_\ell$ such that $x_{q_\ell, p_\ell} \neq 0$ and $\prod_{\ell =
    1}^k x_{q_\ell, p_\ell} \neq 0$ by zero-divisor freeness, we
  conclude that $\nu_k(\sigma)_{\word p1k, j} = 0$ for every $\seq p1k
  \in P \setminus J$, which again proves the claim.  Consequently, all
  states of~$J$ are unreachable.  Since $N$ is trim, we conclude $J =
  \emptyset$, and thus, $X$~has no column of zeroes.

  If $X$~is functional, then it clearly has no row of zeroes.  To
  prove that $X$ has no row of zeroes in the remaining case, let $I =
  \{ q \in Q \mid \forall p \in P \colon x_{qp} = 0\}$.  Then $F_i =
  0$ and $\mu_k(\sigma)_{\word q1k, q} = 0$ for every $\sigma \in
  \Sigma_k$, $q \in Q \setminus I$, $\seq q1k \in Q$, and $i \in I$
  such that $q_\ell = i$ for some $\ell \in [k]$.  Clearly, $F_i =
  \sum_{p \in P} x_{ip} G_p = 0$ for every $i \in I$.  Moreover, since
  $\mu_k(\sigma)X = X^{k, \mathord{\otimes}} \cdot \nu_k(\sigma)$ we
  obtain
  \begin{align}
    \label{eq:Trim2}
    \sum_{q \in Q} \mu_k(\sigma)_{\word q1k, q} \cdot x_{qp} &=
    \sum_{\seq p1k \in P} \Bigl( \prod_{\ell = 1}^k x_{q_\ell,
      p_\ell} \Bigr) \cdot \nu_k(\sigma)_{\word p1k, p} = 0
  \end{align}
  for every $\seq q1k \in Q$, $p \in P$, and $i \in I$ such that
  $q_\ell = i$ for some $\ell \in [k]$.  Since ${\cal A}$ is positive,
  \eqref{eq:Trim2}~implies that $\mu_k(\sigma)_{\word q1k, q} \cdot
  x_{qp} = 0$ for every $q \in Q$.  However, for all $q \in Q
  \setminus I$, there exists $p \in P$ such that $x_{qp} \neq 0$
  because $q \notin I$.  Consequently, $\mu_k(\sigma)_{\word q1k, q} =
  0$ by zero-divisor freeness, which proves the claim.  Thus, all
  states of~$I$ are unreachable.  Since $M$ is trim, we conclude $I =
  \emptyset$, and thus, $X$~has no row of zeroes. \qed
\end{proof}

\begin{definition}[{\protect{see~\cite[Def.~1]{hogmalmay07d}}}]
  \label{df:ForSim}
  Let $M = (\Sigma, Q, \mu, F)$ and $N = (\Sigma, P, \nu, G)$ be
  wta. A surjective function~$\rho \colon Q \to P$ is a \emph{forward
    simulation} from~$M$ to~$N$ if
  \begin{itemize}
  \item[(i)] $F_q = G_{\rho(q)}$ for every $q \in Q$, and
  \item[(ii)] for every $p \in P$, $\sigma \in \Sigma_k$, and $\seq q1k
    \in Q$
    \[ \sum_{q \in Q \colon \rho(q) = p} \mu_k(\sigma)_{\word
      q1k, q} = \nu_k(\sigma)_{\rho(q_1) \dotsm \rho(q_k), p}
    \enspace. \]
  \end{itemize}
  Finally, we say that \emph{$M$~forward simulates~$N$}, written $M
  \twoheadrightarrow N$, if there exists a forward simulation from~$M$
  to~$N$.
\end{definition}

\begin{lemma}
  \label{lm:FSim}
  Let $M$~and~$N$ be wta such that $N$~is trim.  Then $M
  \twoheadrightarrow N$ if and only if there exists a functional
  transfer matrix~$X$ such that $M \stackrel X\to N$.
\end{lemma}

\begin{proof}
  Let $M = (\Sigma, Q, \mu, F)$ and $N = (\Sigma, P, \nu, G)$.  First
  suppose that $M \stackrel X\to N$ with functional~$X \in A^{Q \times
    P}$. Then $\rho_X \colon Q \to P$ is a surjective function by
  Lemma~\ref{lm:Trim}.  Conversely, if $M \twoheadrightarrow N$ with
  the forward simulation $\rho \colon Q \to P$, then $\rho$ induces a
  surjective functional matrix $X \in A^{Q \times P}$ such that
  $\rho_X = \rho$.

  Let $X \in A^{Q \times P}$ be a surjective, functional matrix.  It
  remains to prove that the conditions that (1)~$X$~is a transfer
  matrix and (2)~$\rho_X$ is a forward simulation are equivalent.  We
  discuss the two items of Definitions
  \ref{df:Conj}~and~\ref{df:ForSim} separately.
  \begin{itemize}
  \item[(i)] $F = XG$ if and only if $F_q = G_{\rho(q)}$ for every $q
    \in Q$.
  \item[(ii)] for every $\sigma \in \Sigma_k$, $\seq q1k \in Q$, and
    $p \in P$
    \begin{align*}
      (\mu_k(\sigma)X)_{\word q1k, p} &= \sum_{q \in Q \colon
        \rho_X(q) = p} \mu_k(\sigma)_{\word q1k, q} \\
      (X^{k, \mathord{\otimes}} \cdot \nu_k(\sigma))_{\word q1k, p}
      &= \nu_k(\sigma)_{\rho_X(q_1) \dotsm \rho_X(q_k), p} \enspace.
    \end{align*}
  \end{itemize}
  Thus, $X$~is a transfer matrix if and only if $\rho_X$~is a forward
  simulation, which proves the statement. \qed
\end{proof}

\begin{definition}[{\protect{see~\cite[Def.~16]{hogmalmay07d}}}]
  \label{df:BackSim}
  Let $M = (\Sigma, Q, \mu, F)$ and $N = (\Sigma, P, \nu, G)$ be
  wta.  A surjective function~$\rho \colon Q \to P$ is a \emph{backward
  simulation} from~$M$ to~$N$ if
\begin{itemize}
  \item[(i)] $\sum_{q \in Q \colon \rho(q) = p} F_q = G_p$ for every
    $p \in P$, and
  \item[(ii)] for every $q \in Q$, $\sigma \in \Sigma_k$, and $\seq p1k
    \in P$
    \[ \sum_{\substack{\seq q1k \in Q \\ \rho(q_1) = p_1, \dotsc,
        \rho(q_k) = p_k}} \mu_k(\sigma)_{\word q1k, q} =
    \nu_k(\sigma)_{\word p1k, \rho(q)} \enspace. \]
  \end{itemize}
  Finally, we say that \emph{$M$~backward simulates~$N$}, written $M
  \twoheadleftarrow N$, if there exists a backward simulation from~$M$
  to~$N$.
\end{definition}

\begin{lemma}
  \label{lm:BSim}
  Let $M$~and~$N$ be wta such that $N$ is trim.  Then $M
  \twoheadleftarrow N$ if and only if there exists a transfer
  matrix~$X$ such that $X^{\mathrm T}$ is functional and $N \stackrel
  X\to M$.
\end{lemma}

\begin{proof}
  Let $M = (\Sigma, Q, \mu, F)$ and $N = (\Sigma, P, \nu, G)$.  First,
  suppose that $N \stackrel X\to M$ with the transfer matrix~$X \in
  A^{P \times Q}$ such that $X^{\mathrm T}$~is functional.  Let $Y =
  X^{\mathrm T}$. Then $\rho_Y \colon Q \to P$ is a surjective
  function by Lemma~\ref{lm:Trim}.  Conversely, if $M
  \twoheadleftarrow N$ with the backward simulation $\rho \colon Q \to
  P$, then $\rho$~again induces a surjective, functional matrix $X \in
  A^{Q \times P}$ such that $\rho_X = \rho$.

  Let $X \in A^{Q \times P}$ be a surjective, functional matrix.  It
  remains to prove that the conditions that (1)~$X^{\mathrm T}$~is a
  transfer matrix and (2)~$\rho_X$ is a backward simulation are
  equivalent.  We discuss the two items of Definitions
  \ref{df:Conj}~and~\ref{df:BackSim} separately.
  \begin{itemize}
  \item[(i)] $G = X^{\mathrm T}F$ if and only if $G_p = \sum_{q \in
      Q \colon \rho_X(q) = p} F_q$ for every $p \in P$.
  \item[(ii)] for every $\sigma \in \Sigma_k$, $\seq p1k \in P$, and
    $q \in Q$
    \begin{align*}
      (\nu_k(\sigma)X^{\mathrm T})_{\word p1k, q} &=
      \nu_k(\sigma)_{\word p1k, \rho_X(q)} \\
      ((X^{\mathrm T})^{k, \mathord{\otimes}} \cdot
      \mu_k(\sigma))_{\word p1k, q} &= \sum_{\substack{\seq q1k \in
          Q \\ \rho_X(q_1) = p_1, \dotsc, \rho_X(q_k) = p_k}}
      \mu_k(\sigma)_{\word q1k, q} \enspace.
    \end{align*}
  \end{itemize}
  Thus, $X^{\mathrm T}$~is a transfer matrix if and only if
  $\rho_X$~is a backward simulation, which proves the statement. \qed
\end{proof}

\begin{lemma}
  \label{lm:Help2}
  If $A = \langle U \rangle_{\mathord{+}}$, then for every $X \in A^{Q
    \times P}$ there exist matrices $C, E, D$ such that 
  \begin{itemize}
  \item $X = CED$,
  \item $C^{\mathrm T}$~and~$D$ are functional, and
  \item $E$~is an invertible diagonal matrix.
  \end{itemize}
  If (i)~$X$ is nondegenerate or (ii)~${\cal A}$ has (nontrivial)
  zero-sums, then $C^{\mathrm T}$~and~$D$ can be chosen to be
  surjective.
\end{lemma}

\begin{proof}
  For every $q \in Q$ and $p \in P$, let $\ell_{qp} \in \nat$ and
  $u_{qp1}, \dotsc, u_{qp\ell_{qp}} \in U$ be such that $x_{qp} =
  \sum_{i = 1}^{\ell_{qp}} u_{qpi}$.  In addition, let
  \[ J = \{ (q, i, p) \mid q \in Q, p \in P, i \in [\ell_{qp}] \}
  \enspace. \] Finally, let $\pi_1 \colon J \to Q$ and $\pi_3 \colon J
  \to P$ be such that $\pi_1(\langle q, i, p\rangle) = q$ and
  $\pi_3(\langle q, i, p\rangle) = p$ for every $\langle q, i,
  p\rangle \in J$.  Then we set $C^{\mathrm T}$~and~$D$ to the
  functional matrices represented by $\pi_1$~and~$\pi_3$,
  respectively.  Together with the diagonal matrix~$E$ such that
  $e_{\langle q, i, p\rangle, \langle q, i, p\rangle} = u_{qpi}$ for
  every $\langle q, i, p\rangle \in J$, we obtain $X = CED$.  For
  every $q \in Q$ and $p \in P$ we have
  \[ \sum_{j_1, j_2 \in J} c_{q, j_1} e_{j_1, j_2} d_{j_2, p} =
  \sum_{i = 1}^{\ell_{qp}} e_{\langle q, i, p\rangle, \langle q,
    i, p\rangle} = \sum_{i = 1}^{\ell_{qp}} u_{qpi} = x_{qp}
  \enspace. \]

  It is clear that $C^{\mathrm T}$~and~$D$ are functional matrices.
  Moreover, $E$~is an invertible diagonal matrix because $EE^{-1} = I
  = E^{-1}E$ where $E^{-1}$ is the matrix obtained from~$E$ by
  inverting each nonzero element.  If $X$~is nondegenerate, then
  $C^{\mathrm T}$~and~$D$ are surjective.  Finally, if there are
  zero-sums, then for every $q \in Q$ and $p \in P$ there exist $u, v
  \in U$ such that $x_{qp} = 0 = u + v$, which yields that we can
  choose $\ell_{qp} > 0$.  This completes the proof. \qed
\end{proof}

\begin{lemma}
  \label{lm:Sol}
  Let ${\cal A}$~be equisubtractive.  Moreover, let $R \in A^Q$ and $C
  \in A^P$ be such that $\sum_{q \in Q} r_q = \sum_{p \in P} c_p$.
  Then there exists a matrix $X \in A^{Q \times P}$ with row sums~$R$
  and column sums~$C$; i.e., $\sum_{q \in Q} x_{qp} = c_p$ for every
  $p \in P$ and $\sum_{p \in P} x_{qp} = r_q$ for every $q \in Q$.
\end{lemma}

\begin{proof}
  If $\abs Q \leq 1$ or $\abs P \leq 1$, then the statement is
  trivially true.  Otherwise, select $i \in Q$ and $j \in P$, and let
  $Q' = Q \setminus \{i\}$ and $P' = P \setminus \{j\}$.  By assumption
  \[ \sum_{q \in Q'} r_q + r_i = \sum_{p \in P'} c_p + c_j
  \enspace. \] Thus, by equisubtractivity there exist $a, c'_j, r'_i,
  x_{ij} \in A$ such that \[ \sum_{q \in Q'} r_q = a + c'_j \qquad r_i
  = r'_i + x_{ij} \qquad \sum_{p \in P'} c_p = a + r'_i \qquad c_j =
  c'_j + x_{ij} \enspace. \] Continuing the row decomposition, we
  obtain $Y \in A^{Q'}$ and $R' \in A^{Q'}$ such that $r_q = r'_q +
  y_q$ for every $q \in Q'$ and $\sum_{q \in Q'} r'_q = a$.  In a
  similar manner we perform column decomposition to obtain $Y' \in
  A^{P'}$ and $C' \in A^{P'}$ such that $c_p = c'_p + y'_p$ for every
  $p \in P'$ and $\sum_{p \in P'} c'_p = a$.  Thus, by the induction
  hypothesis, there exists a matrix $X' \in A^{Q' \times P'}$ with row
  sums~$R'$ and column sums~$C'$ because $\sum_{q \in Q'} r'_q =
  \sum_{p \in P'} c'_p$.  Then the matrix \[ X =
  \begin{pmatrix}
    \; & & \; & \\
    & X' & & Y \\
    & & & \\
    & (Y')^{\mathrm T} & & x_{ij}
  \end{pmatrix}
  \] obviously has the required row and column sums $R$~and~$C$,
  respectively. \qed
\end{proof}

\begin{lemma}
  \label{lm:3}
  If $X \in A^{Q \times P}$~is functional (respectively, invertible
  diagonal), then $X^{k, \mathord{\otimes}}$ is functional
  (respectively, invertible diagonal) for every $k \in \nat$.
\end{lemma}

\begin{proof}
  Trivial. \qed
\end{proof}

\begin{theorem}
  \label{thm:2}
  Let $M$~and~$N$ be wta and ${\cal A}$~be equisubtractive with $A =
  \langle U \rangle_{\mathord{+}}$.  Then $M \stackrel X\to N$ if and
  only if there exist two wta $M'$~and~$N'$ such that
  \begin{itemize}
  \item $M \stackrel C\to M'$ where $C^{\mathrm T}$~is functional,
  \item $M' \stackrel E\to N'$ where $E$~is an invertible diagonal
    matrix, and
  \item $N' \stackrel D\to N$ where $D$~is functional.
  \end{itemize}
  If $M$~and~$N$ are trim, then $M' \twoheadleftarrow M$ and $N'
  \twoheadrightarrow N$.
\end{theorem}

\begin{proof}
  Clearly, $M \stackrel C\to M' \stackrel E\to N' \stackrel D\to N$,
  which proves that $M \stackrel{CED}\longrightarrow N$.  For the
  converse, let $M = (\Sigma, Q, \mu, F)$ and $N = (\Sigma, P, \nu,
  G)$. Lemma~\ref{lm:Help2} shows that there exist matrices $C, E, D$
  such that
  \begin{itemize}
  \item $X = CED$,
  \item $C^{\mathrm T}$~and~$D$ are functional matrices, and
  \item $E \in A^{I \times I}$~is an invertible diagonal matrix.
  \end{itemize}
  Finally, let $\varphi \colon I \to Q$ and $\psi \colon I \to P$ be
  the functions associated to $C^{\mathrm T}$~and~$D$.  It remains to
  determine the wta $M'$~and~$N'$.  We construct $M' = (\Sigma, I,
  \mu', F')$ and $N' = (\Sigma, I, \nu', G')$ with
  \begin{itemize}
  \item $G' = DG$ and
  \item $F' = EDG$.
  \end{itemize}
  Then $CF' = CEDG = XG = F$.  Thus, it remains to specify
  $\mu'_k(\sigma)$~and~$\nu'_k(\sigma)$ for every $\sigma \in
  \Sigma_k$.  To this end, we determine a matrix $Y \in A^{I^k \times
    I}$ such that
  \begin{align}
    \label{eq:2a}
    C^{k, \mathord{\otimes}} \cdot Y &= \mu_k(\sigma) CE \\
    \label{eq:2b}
    YD &= E^{k, \mathord{\otimes}} \cdot D^{k, \mathord{\otimes}}
    \cdot \nu_k(\sigma) \enspace.
  \end{align}
  Given such a matrix~$Y$, we then let $\mu'_k(\sigma) = YE^{-1}$ and
  $\nu'_k(\sigma) = (E^{k, \mathord{\otimes}})^{-1} \cdot Y$.  Then
  \begin{align*}
    \mu_k(\sigma) C &= C^{k, \mathord{\otimes}} \cdot \mu'_k(\sigma)
    &\quad \mu'_k(\sigma) E &= E^{k, \mathord{\otimes}} \cdot
    \nu'_k(\sigma) &\quad \nu'_k(\sigma) D &= D^{k,
      \mathord{\otimes}} \cdot \nu_k(\sigma) \enspace.
  \end{align*}
  These equalities are displayed in Fig.~\ref{fig:Squares}.

  Finally, we need to specify the matrix~$Y$.  For every $q \in Q$ and
  $p \in P$, let $I_q = \varphi^{-1}(q)$ and $J_p = \psi^{-1}(p)$.
  Obviously, $Y$~can be decomposed into disjoint (not necessarily
  contiguous) submatrices $Y_{\word q1k, p} \in A^{(I_{q_1} \times
    \dotsm \times I_{q_k}) \times J_p}$ with $\seq q1k \in Q$ and $p
  \in P$.  Then \eqref{eq:2a}~and~\eqref{eq:2b} hold if and only if
  for every $\seq q1k \in Q$ and $p \in P$ the following two
  conditions hold:
  \begin{enumerate}
  \item For every $i \in I$ such that $\psi(i) = p$, the sum of the
    $i$-column of~$Y_{\word q1k, p}$ is $\mu_k(\sigma)_{\word q1k,
      \varphi(i)} \cdot e_{i,i}$.
  \item For all $\seq i1k \in I$ such that $\varphi(i_j) = q_j$ for
    every $j \in [k]$, the sum of the $(\seq i1k)$-row of~$Y_{\word
      q1k, p}$ is $\prod_{j = 1}^k e_{i_j, i_j} \cdot
    \nu_k(\sigma)_{\psi(i_1) \dotsm \psi(i_k), p}$.
  \end{enumerate}
  Those two conditions are compatible because
  \begin{align*}
    &\phantom{{}={}} \sum_{\substack{i \in I \\ \psi(i) = p}}
    \mu_k(\sigma)_{\word q1k, \varphi(i)} \cdot e_{i,i} = \bigl(
    \mu_k(\sigma)CED \bigr)_{\word q1k, p} = \bigl( \mu_k(\sigma)X
    \bigr)_{\word q1k, p} \\
    &\stackrel\dagger= \bigl( X^{k, \mathord{\otimes}} \cdot
    \nu_k(\sigma) \bigr)_{\word q1k, p} = \bigl( C^{k,
      \mathord{\otimes}} \cdot E^{k, \mathord{\otimes}} \cdot D^{k,
      \mathord{\otimes}} \cdot
    \nu_k(\sigma) \bigr)_{\word q1k, p} \\
    &= \sum_{\substack{\seq i1k \in I \\ \forall j \in [k] \colon
        \varphi(i_j) = q_j}} \Bigl( \prod_{j = 1}^k e_{i_j, i_j}
    \Bigr) \cdot \nu_k(\sigma)_{\psi(i_1) \dotsm \psi(i_k), p}
    \enspace.
  \end{align*}
  Consequently, the row and column sums of the submatrices~$Y_{\word
    q1k, p}$ are consistent, which yields that we can determine all
  the submatrices (and thus the whole matrix) by Lemma~\ref{lm:Sol}.

  If $M$~and~$N$ are trim, then either
  \begin{itemize}
  \item[(a)] ${\cal A}$~is zero-sum free (and thus positive because it
    is additively generated by its units), in which case $X$~is
    nondegenerate by Lemma~\ref{lm:Trim}, or
  \item[(b)] ${\cal A}$~has nontrivial zero-sums.
  \end{itemize}
  In both cases, Lemma~\ref{lm:Help2} shows that the matrices
  $C^{\mathrm T}$~and~$D$ are surjective, which yields the additional
  statement by Lemmata \ref{lm:FSim}~and~\ref{lm:BSim}. \qed
\end{proof}

\begin{figure}
  \centering
  \includegraphics{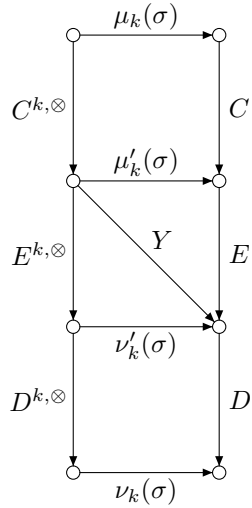}
  \caption{Illustration of the relations between the matrices in the
    proof of Theorem~\protect{\ref{thm:2}}.}
  \label{fig:Squares}
\end{figure}

\section{Category of simulations}
\label{sec:Category}
In this section our aim is to show that several well-known
constructions of wta are \emph{functorial}: they may be extended to
simulations in a functorial way.  Below we will only deal with the
sum, \textsc{Hadamard} product, $\sigma_0$-product, and
$\sigma_0$-iteration (cf.~\cite{esi10}).  Scalar OI-substition,
${}^\dagger$~\cite{bloesi03}, homomorphism, quotient, and
top-concatenation~\cite{esi10} may be covered in a similar fashion.

Throughout this section, let $\mathcal{A}$~be commutative.  Let $M =
(\Sigma, Q, \mu, F)$, $M' = (\Sigma, Q', \mu', F')$, and $M'' =
(\Sigma, Q'', \mu'', F'')$ be wta.  We already remarked that, if $M
\stackrel X\to M'$ and $M' \stackrel Y\to M''$, then $M
\stackrel{XY}\to M''$.  Moreover, $M \stackrel I\to M$ with the unit
matrix~$I \in A^{Q \times Q}$.  Thus, wta over the alphabet~$\Sigma$
form a category~$\text{\textbf{Sim}}_\Sigma$.

In the following, let $M = (\Sigma, Q, \mu, F)$ and $N = (\Sigma, P,
\nu, G)$ be wta such that $Q \cap P = \emptyset$.

\begin{definition} 
  The sum $M + N$ of $M$~and~$N$ is the wta $(\Sigma, Q \cup P,
  \kappa, H)$ where $H = \langle F, G \rangle = \begin{pmatrix} F \\
    G \end{pmatrix}$ and
  \[ \kappa_k(\sigma)_{\word q1k, q} = (\mu_k(\sigma) +
  \nu_k(\sigma))_{\word q1k, q} =
  \begin{cases}
    \mu_k(\sigma)_{\word q1k, q} & \text{if } q, \seq q1k \in Q \\
    \nu_k(\sigma)_{\word q1k, q} & \text{if } q, \seq q1k \in P \\
    0 & \text{otherwise.}
  \end{cases}
  \]
  for all $\sigma \in \Sigma_k$ and $q, \seq q1k \in Q \cup P$.
\end{definition} 

It is well-known that $\sem{M + N} = \sem M + \sem N$.  Next, we
extend the sum construction to simulations.  To this end, let $M
\stackrel X\to M'$ with $M' = (\Sigma, Q', \mu', F')$, and let $N
\stackrel Y\to N'$ with $N' = (\Sigma, P', \nu', G')$.

\begin{definition}
  The sum $X + Y \in A^{(Q \cup P) \times (Q' \cup P')}$ of the
  transfer matrices $X$~and~$Y$ is
  \[ X + Y =
  \begin{pmatrix} 
    X & 0 \\
    0 & Y 
  \end{pmatrix} \enspace.
  \]
\end{definition} 

\begin{proposition}
  We have~$(M + N) \stackrel{X + Y}\longrightarrow (M' + N')$.
\end{proposition}

\begin{proof}
  We only need to verify the two conditions of
  Definition~\ref{df:Conj}.  For every $\sigma \in \Sigma_k$ we have
  \begin{align*}
    &\phantom{{}={}} \bigl(\mu_k(\sigma) + \nu_k(\sigma) \bigr) \cdot
    (X + Y) = \mu_k(\sigma)X + \nu_k(\sigma)Y \\
    &= X^{k, \mathord{\otimes}} \cdot \mu'_k(\sigma) + Y^{k,
      \mathord{\otimes}} \cdot \mu'_k(\sigma) = (X + Y)^{k,
      \mathord{\otimes}} \cdot \bigl(\mu'_k(\sigma) + \nu'_k(\sigma)
    \bigr)
  \end{align*}
  and $\langle F, G \rangle = \langle XF', YG' \rangle = (X + Y) \cdot
  \langle F', G' \rangle$, which completes the proof. \qed
\end{proof}

\begin{proposition}
  The function~$+$, which is defined on wta and transfer matrices, is
  a functor $\text{\textbf{Sim}}_\Sigma^2 \to
  \text{\textbf{Sim}}_\Sigma$.
\end{proposition}

\begin{proof}
  It is a routine matter to verify that identity transfer matrices are
  preserved and $(X + Y) \cdot (X' + Y') = XX' + YY'$ for all
  composable transfer matrices~$X, X', Y, Y'$. \qed
\end{proof}

\begin{definition}
  Let $\sigma_0$ be a distinguished symbol in $\Sigma_0$.  The
  $\sigma_0$-product $M \cdot_{\sigma_0} N$ of~$M$ with~$N$ is the wta
  $(\Sigma, Q \cup P, \kappa, H)$ such that
  \[ H = \langle F, 0\rangle = \begin{pmatrix} F\\ 0 \end{pmatrix} \]
  and for each $\sigma \in \Sigma_k$ with $\sigma \neq \sigma_0$,
  \[ \kappa_k(\sigma)_{\word q1k, q} = 
  \begin{cases}
    \mu_k(\sigma)_{\word q1k, q} & \text{if } q, \seq q1k \in Q \\
    \mu_0(\sigma_0)_q \cdot \sum_{p \in P} \nu_k(\sigma)_{\word q1k,
      p} G_p & \text{if } q \in Q \text{ and } \seq q1k \in P \\ 
    \nu_k(\sigma)_{\word q1k, q} & \text{if } q, \seq q1k \in P \\
    0 & \text{otherwise.}
  \end{cases} 
  \]
  Moreover,
  \[ \kappa_0(\sigma_0)_q = 
  \begin{cases}
    \mu_0(\sigma_0)_q \cdot \sum_{p \in P} \nu_0(\sigma_0)_p G_p &
    \text{if } q\in Q \\
    \nu_0(\sigma_0)_q & \text{if } q \in P.
  \end{cases}
  \]
\end{definition}

It is known that $\sem{M \cdot_{\sigma_0} N} = \sem M \cdot_{\sigma_0}
\sem N$.  We extend this construction to simulations.  To this end,
let $M \stackrel X\to M'$ and $N \stackrel Y\to N'$.  Then we define
$X \cdot_{\sigma_0} Y = X + Y$.  The next proposition can be verified
by a routine calculation.

\begin{proposition}
  The function~$\cdot_{\sigma_0}$, which is defined on wta and
  transfer matrices, is a functor $\text{\textbf{Sim}}_\Sigma^2 \to
  \text{\textbf{Sim}}_\Sigma$.
\end{proposition}

\begin{definition}
  The \textsc{Hadamard} product $M \cdot_{\mathrm H} N$ is the wta
  $(\Sigma, Q \times P, \kappa, H)$ where $H = F \otimes G$ and
  $\kappa_k(\sigma) = \mu_k(\sigma) \otimes \nu_k(\sigma)$ for all
  $\sigma \in \Sigma_k$.
\end{definition}

We again extend the construction to simulations.  If $M \stackrel X\to
M'$ and $N \stackrel Y\to N'$, then we define $X \cdot_{\mathrm H}
X \otimes Y$.

\begin{proposition}
  The function~$\cdot_{\mathrm H}$, which is defined on wta and
  transfer matrices, is a functor $\text{\textbf{Sim}}_\Sigma^2 \to
  \text{\textbf{Sim}}_\Sigma$.
\end{proposition}

Finally, we deal with iteration. Let $\sigma_0$~be a fixed symbol
in~$\Sigma_0$.  Here we assume that $\mathcal{A}$~is complete.  Thus,
$\mathcal{A}$~comes with a star operation $a^* = \sum_{n \in \nat}
a^n$ for every $a \in A$.

\begin{definition}
  The $\sigma_0$-iteration $M^{*_{\sigma_0}}$ of~$M$ is the wta
  $(\Sigma, Q, \kappa, F)$ where
  \[ \kappa_k(\sigma)_{\word q1k, q} = \mu_k(\sigma)_{\word q1k, q} +
  \sem{M}(\sigma_0)^* \cdot \sum_{p \in Q} \mu_k(\sigma)_{\word q1k,
    p} F_p \] for all $\sigma \in \Sigma_k \setminus \{\sigma_0\}$ and
  $\kappa_0(\sigma_0) = \mu_0(\sigma_0)$.
\end{definition}

If $M \stackrel X\to M'$, then we define $X^{*_{\sigma_0}} = X$.

\begin{proposition}
  The $\sigma_0$-iteration, which is defined on wta and transfer
  matrices, is a functor $\text{\textbf{Sim}}_\Sigma \to
  \text{\textbf{Sim}}_\Sigma$.
\end{proposition}


\begin{remark}
  Several subcategories of~$\text{\textbf{Sim}}_\Sigma$ are also of
  interest, for example the categories formed by the relational or
  functional simulations and their duals.  The above
  constructions are preserved by these special kinds of simulations.
\end{remark} 

\section{Joint reduction}
\label{sec:Joint}
Next we will establish equivalence results using the approach called
\emph{joint reduction} in~\cite{bealomsak06}.  Let $V \subseteq A^I$
be a set of vectors for a finite set~$I$.  Then the ${\cal
  A}$-semimodule generated by~$V$ is denoted by~$\langle V \rangle$.
Given two wta $M = (\Sigma, Q, \mu, F)$ and $N = (\Sigma, P, \nu, G)$
with $Q \cap P = \emptyset$, we first compute $M + N = (\Sigma, Q \cup
P, \mu', F')$ as defined in Section~\ref{sec:Category}.  Now the aim
is to compute a finite set~$V \subseteq A^{Q \cup P}$ such that
\begin{itemize}
\item[(i)] $(v_1 \otimes \dotsm \otimes v_k) \cdot \mu'_k(\sigma) \in
  \langle V \rangle$ for every $\sigma \in \Sigma_k$ and $\seq v1k \in
  V$, and
\item[(ii)] $v_1F = v_2G$ for every $(v_1, v_2) \in V$ such that $v_1
  \in A^Q$ and $v_2 \in A^P$.
\end{itemize}

With such a finite set~$V$ we can now construct a wta $M' = (\Sigma,
V, \nu', G')$ with $G'_v = v F'$ for every $v \in V$ and
\[ \sum_{v \in V} \nu'_k(\sigma)_{\word v1k, v} \cdot v = (v_1 \otimes
\dotsm \otimes v_k) \cdot \mu'_k(\sigma) \] for every $\sigma \in
\Sigma_k$ and $\seq v1k \in V$.  It remains to prove that $M'$
simulates~$M + N$.  To this end, let $X = (v)_{v \in V}$, where each
$v \in V$ is a row vector.  Then for every $\sigma \in \Sigma_k$,
$\seq v1k \in V$, and $q \in Q \cup P$, we have
\begin{align*}
  &\phantom{{}={}} (\nu'_k(\sigma) X)_{\word v1k, q} = \sum_{v \in V}
  \nu'_k(\sigma)_{\word v1k, v} \cdot v_q = \Bigl( \sum_{v \in V}
  \nu'_k(\sigma)_{\word v1k, v} \cdot v \Bigr)_q \\
  &= \bigl( (v_1 \otimes \dotsm \otimes v_k) \cdot \mu'_k(\sigma)
  \bigr)_q = \sum_{\seq q1k \in Q \cup P} (v_1)_{q_1} \cdot \ldots
  \cdot (v_k)_{q_k} \cdot \mu'_k(\sigma)_{\word q1k, q} \\
  &= \bigl( X^{k, \mathord{\otimes}} \cdot \mu'_k(\sigma)
  \bigr)_{\word v1k, q} \enspace.
\end{align*}
Moreover, if we let $X_1$~and~$X_2$ be the restrictions of~$X$ to the
entries of $Q$~and~$P$, respectively, then we have $\nu'_k(\sigma)X_1
= X_1^{k, \mathord{\otimes}} \cdot \mu_k(\sigma)$ and
$\nu'_k(\sigma)X_2 = X_2^{k, \mathord{\otimes}} \cdot \nu_k(\sigma)$.
In addition, $G'_v = v F' = \sum_{q \in Q \cup P} v_q F'_q = (XF')_v$
for every $v \in V$, which proves that $M' \stackrel X\to (M + N)$.
Since $v_1F = v_2G$ for every $(v_1, v_2) \in V$, we have $G'_{(v_1,
  v_2)} = (v_1, v_2)F' = v_1F + v_2G = (1+1)v_1F = (1 + 1)v_2G$.  Now,
let $G''_{(v_1, v_2)} = v_1F = v_2G$ for every $(v_1, v_2) \in V$.
Then
\begin{align*}
  G''_v &= v_1F = \sum_{q \in Q} v_q F_q = (X_1F)_v  \\
  &= v_2G = \sum_{p \in P} v_p G_p = (X_2G)_v 
\end{align*} 
for every $v = (v_1, v_2) \in V$.  Consequently, $M'' \stackrel{X_1}\to
M$ and $M'' \stackrel{X_2}\to N$, where $M'' = (\Sigma, V, \nu',
G'')$.  This proves the next theorem.

\begin{theorem}
  \label{thm:Joint}
  Let $M$~and~$N$ be two equivalent wta.  If there exists a finite
  set~$V \subseteq A^{Q \cup P}$ with properties (i)~and~(ii), then
  there exists a chain of simulations that join $M$~and~$N$.  In fact,
  there exists a single wta that simulates both $M$~and~$N$.
\end{theorem}

\subsection{Fields}
\label{sec:Fields}
In this section, let ${\cal A}$~be a field.  We first recall some
notions from~\cite{aleboz89}.  Let $\varphi \in \series A\Sigma$ be a
tree series.  The \emph{syntactic ideal} of~$\varphi$ is
\[ I_\varphi = \{ \psi \in \series A\Sigma \mid \sum_{t \in T_\Sigma}
(\psi, t) (\varphi, c[t]) = 0 \text{ for all } c \in C_\Sigma \}
\enspace. \] Moreover, let $\mathord\equiv$ be the equivalence
relation on~$\series A\Sigma$ such that $\psi_1 \equiv \psi_2$ if and
only if $\psi_1 - \psi_2 \in I_\varphi$.  The syntactic algebra is
$[\series A\Sigma]_\equiv$.  By~\cite[Proposition~2]{aleboz89} the
tree series~$\varphi$ is recognizable if and only if its syntactic
algebra has finite dimension.  Now, let $\varphi$~be recognizable, and
let $B$~be a basis of its syntactic algebra.  Finally, let
$M_\varphi$~be the obtained canonical weighted tree automaton, which
recognizes~$\varphi$.

\begin{theorem}[{\protect\cite[p.~453]{aleboz89}}]
  \label{thm:ReltoMin}
  Every trim wta recognizing~$\varphi$ simulates~$M_\varphi$.
\end{theorem}

Consequently, all equivalent trim wta $M_1$~and~$M_2$ simulate the
canonical wta that recognizes~$\sem M$.  Using Theorem~\ref{thm:2} we
can show that there exist wta~$M'_1$, $M'_2$, $N'_1$, and~$N'_2$ such
that
\begin{itemize}
\item $M_1 \twoheadleftarrow M'_1$,
\item $M'_1 \stackrel E\to N'_1$ with an invertible diagonal
  matrix~$E$,
\item $N'_1 \twoheadrightarrow M_\varphi$,
\item $N'_2 \twoheadrightarrow M_\varphi$,
\item $M'_2 \stackrel{E'}\to N'_2$ with an invertible diagonal
  matrix~$E'$, and
\item $M_2 \twoheadleftarrow M_2$.
\end{itemize}
This can be illustrated as follows:
\[ M_1 \xleftarrow{\text{backward}} M'_1
\xrightarrow{\text{diagonal}} N'_1 
\xrightarrow{\text{forward}} M_\varphi
\xleftarrow{\text{forward}} N'_2
\xleftarrow{\text{diagonal}} M'_2
\xrightarrow{\text{backward}} M_2 \]

\begin{theorem}
  \label{thm:Field}
  Every two equivalent trim wta $M$~and~$N$ over the field~${\cal A}$
  can be joined by a chain of simulations.  Moreover, there exists a
  minimal wta~$M_{\sem M}$ such that $M$~and~$N$ both
  simulate~$M_{\sem M}$.
\end{theorem}

We could have obtained a similar theorem with the help of
Theorem~\ref{thm:Joint} because the finite set~$V$ can be obtained as
in~\cite{boz91}.  The approach in the next section will cover this
case.

\subsection{\textsc{Noetherian} semirings}
\label{sec:Int}
Now, let ${\cal A}$ be a \textsc{Noetherian} semiring.  We construct
the finite set~$V$ as follows.  Let $V_0 = \{ \mu'_0(\alpha) \mid
\alpha \in \Sigma_0\}$ and
\[ V_{i + 1} = V_i \cup \bigl( \{ (v_1 \otimes \dotsm \otimes v_k)
\cdot \mu'_k(\sigma) \mid \sigma \in \Sigma_k, \seq v1k \in V_i\}
\setminus \langle V_i\rangle \bigr) \] for every $i \in \nat$.  Then
\[ \{ 0 \} \subseteq \langle V_0 \rangle \subseteq \langle V_1 \rangle
\subseteq \dotsb \subseteq \langle V_k \rangle \subseteq \dotsb {}
\] is stationary after finitely many steps because ${\cal A}$ is
\textsc{Noetherian}.  Thus, let $V = V_k$ for some $k \in \nat$ such
that $\langle V_k \rangle = \langle V_{k+1} \rangle$. Clearly, $V$~is
finite and has property~(i).  Trivially, $V \subseteq \{ h_{\mu'}(t)
\mid t \in T_\Sigma \}$, so let $v \in V$ be such that $v = \sum_{i
  \in I} (h_\mu(t_i), h_\nu(t_i))$ for some finite index set~$I$ and
$t_i \in T_\Sigma$ for every $i \in I$.  Then
\begin{align*}
  \Bigl( \sum_{i \in I} h_\mu(t_i) \Bigr) F = \sum_{i \in I} (\sem M,
  t_i) = \sum_{i \in I} (\sem N, t_i) = \Bigl( \sum_{i \in I}
  h_\nu(t_i) \Bigr) G
\end{align*}
because $\sem M = \sem N$, which proves property~(ii).

\begin{theorem}
  \label{thm:Noeth}
  Let ${\cal A}$ be a \textsc{Noetherian} semiring.  For every two
  equivalent wta $M$~and~$N$ over~${\cal A}$, there exists a chain of
  simulations that join $M$~and~$N$.  In fact, there exists a single 
  wta that simulates both $M$~and~$N$.
\end{theorem}

\begin{proof}
  Follows from Theorem~\ref{thm:Joint}.
\end{proof}

Since $\integer$~forms a \textsc{Noetherian} ring, we obtain the
following corollary.

\begin{corollary}[{\protect{of Theorem~\ref{thm:Noeth}}}]
  \label{cor:Int}
  For every two equivalent wta $M$~and~$N$ over~$\integer$, there
  exists a chain of simulations that join $M$~and~$N$.  In fact, there
  exists a single wta that simulates both $M$~and~$N$.
\end{corollary}

In fact, since $M + N$ uses only finitely many semiring coefficient,
it is sufficient that every finitely generated subsemiring of~${\cal
  A}$ is contained in a \textsc{Noetherian} subsemiring of~${\cal A}$.
Since every finitely generated commutative ring is
\textsc{Noetherian}~\cite[Cor.~IV.2.4 \& Prop.~X.1.4]{lan84}, we
obtain the following corollary.

\begin{corollary}[{\protect{of Theorem~\ref{thm:Noeth}}}]
  \label{cor:Ring}
  For every two equivalent wta $M$~and~$N$ over the commutative
  ring~${\cal A}$, there exists a chain of simulations that join
  $M$~and~$N$.  In fact, there exists a single wta that simulates both
  $M$~and~$N$.
\end{corollary}

\subsection{Natural numbers}
\label{sec:Nat}
Finally, let ${\cal A} = \nat$ be the semiring of natural numbers.  We
compute the finite set~$V \subseteq \nat^{Q \cup P}$ as follows:
\begin{enumerate}
\item Let $V_0 = \{ \mu'_0(\alpha) \mid \alpha \in \Sigma_0\}$ and $i
  = 0$.
\item For every $v, v' \in V_i$ such that $v \leq v'$, replace~$v'$
  by~$v' - v$.
\item Set $V_{i + 1} = V_i \cup \bigl( \{ (v_1 \otimes \dotsm \otimes
  v_k) \cdot \mu'_k(\sigma) \mid \sigma \in \Sigma_k, \seq v1k \in
  V_i\} \setminus \langle V_i\rangle \bigr)$.
\item Until $V_{i + 1} = V_i$, increase~$i$ and repeat step~2.
\end{enumerate}
Clearly, this algorithm terminates since every vector can only be
replaced by a smaller vector in step~2 and step~3 only adds a finite
number of vectors, which after the reduction in step~2 are pairwise
incomparable.  Moreover, property~(i) trivially holds because at
termination $V_{i+1} = V_i$ after step~3.  Consequently, we only need
to prove property~(ii).  To this end, we first prove that $V \subseteq
\langle \{ h_{\mu'}(t) \mid t \in T_\Sigma \} \rangle_{\mathord{+},
  \mathord{-}}$.  This is trivially true after step~1 because
$\mu'_0(\alpha) = h_{\mu'}(\alpha)$ for every $\alpha \in \Sigma_0$.
Clearly, the property is preserved in steps 2~and~3.  Finally,
property~(ii) can now be proved as follows. Let $v \in V$ be such that
$v = \sum_{i \in I_1} (h_\mu(t_i), h_\nu(t_i)) - \sum_{i \in I_2}
(h_\mu(t_i), h_\nu(t_i))$ for some finite index sets $I_1$~and~$I_2$
and $t_i \in T_\Sigma$ for every $i \in I_1 \cup I_2$.  Then
\begin{align*}
  &\phantom{{}={}} \Bigl( \sum_{i \in I_1} h_\mu(t_i) - \sum_{i \in
    I_2} h_\mu(t_i) \Bigr) F = \sum_{i \in I_1} h_\mu(t_i)F - \sum_{i
    \in I_2} h_\mu(t_i)F \\
  &= \sum_{i \in I_1} (\sem M, t_i) - \sum_{i \in I_2} (\sem M, t_i) =
  \sum_{i \in I_1} (\sem N, t_i) - \sum_{i \in I_2} (\sem N, t_i) \\
  &= \sum_{i \in I_1} h_\nu(t_i)G - \sum_{i \in I_2} h_\nu(t_i)G =
  \Bigl( \sum_{i \in I_1} h_\nu(t_i) - \sum_{i \in I_2} h_\nu(t_i)
  \Bigr) G
\end{align*}
because $\sem M = \sem N$.

\begin{corollary}[{\protect{of Theorem~\ref{thm:Joint}}}]
  \label{cor:Nat}
  For every two equivalent wta $M$~and~$N$ over~$\nat$, there exists a
  chain of simulations that join $M$~and~$N$. In fact, there exists a
  single wta that simulates both $M$~and~$N$.
\end{corollary}

For all finitely and effectively presented semirings, Theorems
\ref{thm:Field}~and~\ref{thm:Noeth} and Corollaries
\ref{cor:Ring}~and~\ref{cor:Nat}, also yield decidability of
equivalence for $M$~and~$N$.  Essentially, we run the trivial
semi-decidability test for inequality and a search for the wta the
simulates both $M$~and~$N$ in parallel.  We know that either test will
eventually return, thus deciding whether $M$~and~$N$ are equivalent.
Conversely, if equivalence is undecidable, then simulation cannot
capture equivalence~\cite{esimal10}.

\bibliography{extra}
\bibliographystyle{splncs03}

\end{document}